\documentclass[
  reprint, 
  aps, 
  superscriptaddress, 
  longbibliography,
  amsmath,amssymb, 
  floatfix
]{revtex4-2}
\usepackage{physics}
\usepackage{amsmath, amssymb, amsthm}
\usepackage{graphicx}
\usepackage{subfigure}
\usepackage{enumerate}
\usepackage{ascmac}
\usepackage{fancybox}
\usepackage{fancyhdr}
\usepackage{bm}
\usepackage{siunitx}
\usepackage{url}
\usepackage{framed}
\usepackage{hyperref}
\hypersetup{%
 setpagesize=false,%
 bookmarks=true,%
 bookmarksdepth=tocdepth,%
 bookmarksnumbered=true,%
 colorlinks=false,%
 pdftitle={},%
 pdfsubject={},%
 pdfauthor={},%
 pdfkeywords={}}
\usepackage{tcolorbox}
\usepackage{cleveref}
\usepackage{booktabs}
\usepackage{bbm}
\usepackage{mleftright}
\usepackage{here}

\theoremstyle{plain}
\newtheorem{thm}{Theorem}
\newtheorem*{thm*}{Theorem}

\newtheorem{lemma}{Lemma}
\theoremstyle{definition}

\newtheorem*{rmk*}{Remark}

\newtheorem*{prpt*}{Property}

\begin{document}
\title{Symmetry and Topology of Successive Quantum Feedback Control}
\author{Junxuan Wen}
\affiliation{Department of Applied Physics, The University of Tokyo, 7-3-1 Hongo, Bunkyo-ku, Tokyo 113-8656, Japan}
\author{Zongping Gong}
\affiliation{Department of Applied Physics, The University of Tokyo, 7-3-1 Hongo, Bunkyo-ku, Tokyo 113-8656, Japan}
\author{Takahiro Sagawa}
\affiliation{Department of Applied Physics, The University of Tokyo, 7-3-1 Hongo, Bunkyo-ku, Tokyo 113-8656, Japan}
\affiliation{Quantum-Phase Electronics Center (QPEC), The University of Tokyo,
7-3-1 Hongo, Bunkyo-ku, Tokyo 113-8656, Japan}
\begin{abstract}
    We establish a symmetry classification 
    for a general class of quantum feedback control. For successive feedback control with a non-adaptive sequence of bare measurements (i.e., with positive Kraus operators), we prove that the symmetry classification collapses to the ten-fold AZ${}^\dagger$ classes, specifying the allowed topology of CPTP maps associated with feedback control. We demonstrate that a chiral Maxwell's demon with Gaussian measurement errors exhibits quantized winding numbers. Moreover, for general (non-bare) measurements, we explicitly construct a protocol that falls outside the ten-fold classification. These results broaden and clarify the principles in engineering topological aspects of quantum control robust against disorder and imperfections.
\end{abstract}
\maketitle

\textit{Introduction.---}
Non-Hermitian topology has reshaped our understanding of phases and robustness in various open and dissipative systems \cite{PhysRevX.8.031079,Kawabata_2019,Ashida_2020,Okuma_2023,Bergholtz2021,PhysRevLett.120.146402,PhysRevLett.121.026808,Yao2018,PhysRevB.99.235112}, including active matter \cite{Sone2020NatComm}, ultracold atoms \cite{Ashida2017,Li2019}, 
electronic circuits \cite{Schindler2011,Ezawa2019}, and classical stochastic systems \cite{Murugan2017,Sawada_2024,sawada2024bulkboundarycorrespondenceergodicnonergodic}.
Based on the 38-fold Bernard-LeClair (BL) extension of the conventional ten-fold Altland-Zirnbauer (AZ) symmetry classification \cite{Bernard_2002}, general non-Hermitian phases in arbitrary spatial dimensions and symmetries have been established \cite{Kawabata_2019,PhysRevB.99.235112}. 
It is worth noting that the notions of non-Hermitian topology should naturally apply to dynamical control in open quantum systems \cite{Nakagawa_2025}.

Meanwhile, quantum feedback control has attracted intensive interest from several perspectives. Continuous measurement and feedback \cite{Wiseman_Milburn_2009} has long been central, as it enables state stabilization \cite{Sayrin2011Nature,Vijay2012,CampagneIbarcq2013PRX,PhysRevLett.70.548,Dotsenko2009PRA}, noise suppression \cite{Pinard2000PRA,Cox2016PRL,Wiseman1999Job,Tsang2010PRL}, feedback cooling \cite{Cohadon1999PRL,Poggio2007PRL,PhysRevLett.70.548,5lp2-9sps,kumasaki2025thermodynamicapproachquantumcooling}, and error correction \cite{Ofek2016Nature,NVCenter2016NatComm,Google2023Nature,Ahn2002PRA,vanHandel2005,Denhez_2012}. Moreover, quantum feedback control has intrinsic connections to quantum thermodynamics as illustrated by Maxwell's demon, which clarifies the role of information in thermodynamics both theoretically \cite{PhysRevLett.100.080403,PhysRevE.88.052121,PhysRevLett.122.150603,PhysRevLett.128.170601,kumasaki2025thermodynamicapproachquantumcooling} and experimentally \cite{PhysRevLett.117.240502,Cottet2017PNAS,Masuyama_2018,PhysRevLett.121.030604,PRXQuantum.3.020329,5lp2-9sps}.
However, the topological aspects of quantum feedback control have rarely been studied. A first step toward such a framework for topological quantum feedback control has been presented in Ref. \cite{Nakagawa_2025}, while the results of symmetry classification therein are limited to the case where only a single projective measurement without
error is executable. 

In this Letter, we provide a general symmetry classification of successive feedback control. We prove a no-go theorem stating that, despite the greatly larger protocol space 
compared to the single projective measurement case,
the admissible BL classes again collapse to the AZ${^\dagger}$ classes. Our central assumption is that each measurement is bare \cite{Jacobs_2005}, meaning that all Kraus operators are positive. As a consequence, some symmetries that might appear available in generic non-Hermitian evolutions are ruled out. We illustrate the topology of a chiral Maxwell's demon subject to Gaussian measurement errors, for which we observe quantized winding numbers. Since any POVM admits a realization via a bare measurement, our result broadly applies and provides design principles for robust quantum feedback protocols based on topology and symmetry.

Furthermore, we go beyond the above no-go theorem by dropping the condition that measurements are bare; allowing the unitary parts in the measurements, we explicitly construct a protocol whose symmetry class falls outside AZ${^\dagger}$. This suggests that bareness is the assumption that draws the precise boundary inside which the classification collapses to the AZ${^\dagger}$ classes.

\textit{Setup.---}
We consider quantum feedback control composed of measurement and feedback. In general, the measurement is described by Kraus operators $\{ \hat{M}_m \}$ satisfying $\sum_m \hat{M}_m^\dagger \hat{M}_m = \hat{I}$, and the feedback is described by unitary operators conditioned on the outcomes of preceding measurements. In our setup (see also Fig. 1), the CPTP map for the entire process is written as
\begin{widetext}
    \begin{equation}\label{general form of feedback control}
    \begin{aligned}
        \mathcal{E}_{\tau_1,\cdots,\tau_N}(\hat{\rho}) := 
        \sum_{m_1,\cdots,m_N}^{}\hat{K}^{(N)}_{m_1,\cdots,m_N}(\tau_N)\cdots\hat{K}^{(2)}_{m_1,m_2}(\tau_2)\hat{K}^{(1)}_{m_1}(\tau_1) \hat{\rho} \hat{K}^{(1)\dag}_{m_1}(\tau_1)\hat{K}^{(2)\dag}_{m_1,m_2}(\tau_2)\cdots\hat{K}^{(N)\dag}_{m_1,\cdots,m_N}(\tau_N),
    \end{aligned}
\end{equation}
\end{widetext}
\begin{equation}
    \hat{K}^{(i)}_{m_1,\cdots,m_i}(\tau_i) := \hat{U}^{(i)}_{m_1,\cdots,m_i}(\tau_i)\hat{M}^{(i)}_{m_i}, 
\end{equation}
where $\left\{ \hat{M}_{m_i}^{(i)} \right\}$ is the set of measurement operators at step $i$, and $\hat{U}_{m_1,\cdots,m_i}^{(i)}(\tau_i)$ is the feedback unitary conditioned on outcomes $m_1,\cdots,m_i$ with feedback duration $\tau_i$.

In general, the POVM $\{ \hat E_m \}$ associated with the Kraus operators $\{ \hat M_m \}$ is given by $\hat E_m = \hat M_m^\dagger \hat M_m$.  Then, each Kraus operator has the polar decomposition $\hat M_m = \hat V_m \sqrt{\hat E_m}$ with $\hat V_m$ being some unitary.  There is a certain arbitrariness regarding whether this unitary $\hat V_m$ is regarded as a part of measurement or a part of feedback.  Here, we suppose that the feedback in our setup does not include such unitaries, but is only given by the continuous-time evolution generated by a Hamiltonian, written as 
\begin{equation}\label{unitary with feedback Hamiltonian}
    \hat{U}^{(i)}_{m_1,\cdots,m_i}(\tau_i) := e^{-i\int_{0}^{\tau_i}\hat{H}^{(i)}_{m_1,\cdots,m_i}(t_i)dt_i}.
\end{equation}

In the main theorem (stated later) of this work, we further assume that the unitary $\hat V_m$ is absent in the measurement and thus the Kraus operator is simply given by $\hat M_m = \sqrt{\hat E_m}$ ($\geq 0$).  We refer to this class of measurements as \textit{bare} measurements \cite{Jacobs_2005}. We note that $\hat M_m = \sqrt{\hat E_m}$ is a unique choice making the Kraus operator positive, and it is known that such a class can be characterized as measurements with minimal disturbance\cite{Wiseman_Milburn_2009,barnum2002informationdisturbancetradeoffquantummeasurement}. Bare measurements include, for example, quantum non-demolition (QND) measurements \cite{BraginskyKhalili1996RMPQND}, which have diverse experimental applications \cite{Guerlin2007,PhysRevLett.102.033601,Murch2013Nature12539,Kuzmich2000PRL851594,Lupascu2007NatPhysNphys509,Anand2024}.

In addition, we assume that the sequence of measurements is non-adaptive throughout this paper. That is, whereas feedback is dependent on the preceding measurement outcomes, the measurement itself is fixed in advance and does not depend on preceding outcomes.

\begin{figure}
\centering
\includegraphics[width=1\linewidth]{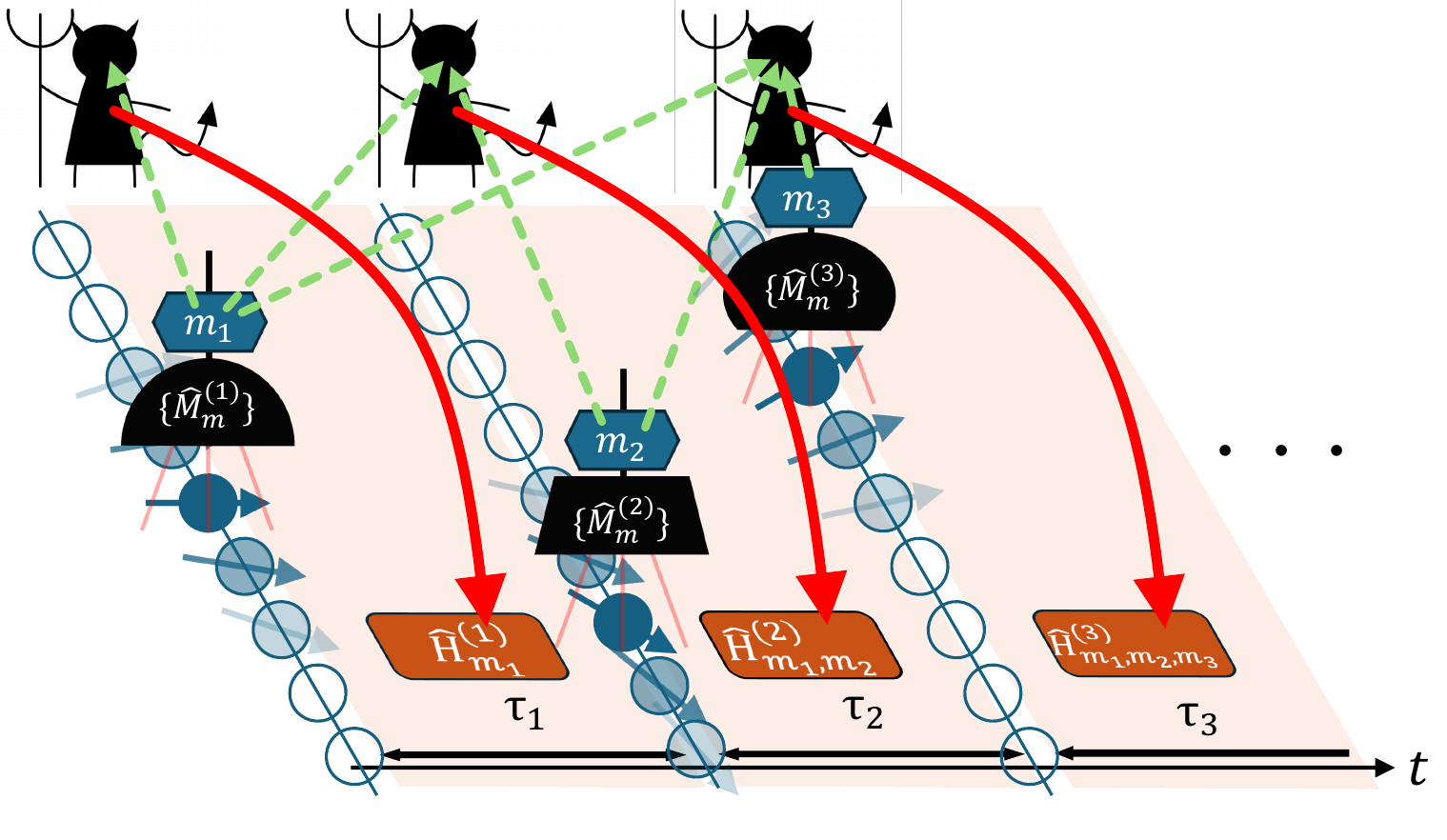}
\caption{Schematic of successive feedback control with non-adaptive measurements. The Kraus operators $\{ \hat{M}_m^{(i)} \}$ of the $i$-th measurement can depend on $i$ but are fixed in advance, and the feedback is applied via feedback Hamiltonians conditioned on previous outcomes.} 
\label{Schematic}
\end{figure}

We now consider a more detailed setup relevant to topological classifications.  Consider a single particle on a $d$-dimensional lattice with translational symmetry. 
Under periodic boundary conditions (PBC), $\mathcal{E}$ can be block-diagonalized in the eigenspaces of the translation operators. Each block corresponds to the Bloch matrix \cite{Nakagawa_2025} defined by 
\begin{equation}\label{Bloch Matrix}
    \begin{aligned}
        & X_{a,b,\bm{\mu};c,d,\bm{\nu}}(\bm{k}) \\
        & := \frac{1}{N_\text{cell}} \sum_{\bm{j,j'}}^{}\sum_{m}^{} (\hat{K}_m)_{\bm{j},a;\bm{j'},c} (\hat{K}_m)^*_{\bm{j+\mu},b;\bm{j'+\nu},d} e^{-i\bm{k}\cdot(\bm{R_j}-\bm{R_{j'}})},
    \end{aligned}
\end{equation}
where $\bm{j,j',\mu,\nu}$ denote the lattice sites, $a,b,c,d$ denote the internal degrees of freedom, $\bm{R}_j$ is the position of site $\bm{j}$, and $N_{\text{cell}}$ is the number of unit cells.
Thus, solving the eigenvalue problem $\mathcal{E}(\hat{\rho}_n) = \xi_n \hat{\rho}_n$ is reduced to that of the Bloch matrices $X(\bm{k})$. In the limit of infinite system size, $\bm{k}$ can be treated as continuous.

\textit{Symmetry classification of successive feedback control with non-adaptive bare measurements.---}
The symmetries of topological insulators are classified by the presence or absence of (i) time reversal symmetry (TRS) $U_T H^* U_T^\dagger = H$, (ii) particle-hole symmetry (PHS) $U_C H^\top U_C^\dagger = -H$, and (iii) chiral symmetry (CS) $U_S H^\dagger U_S^\dagger = -H$, which leads to 10 classes called AZ classes \cite{PhysRevB.55.1142}. 
When the system is non-Hermitian, the transposition and complex conjugation are no longer equivalent. As a result, two symmetries TRS$^\dagger$: $U_T H^\top U_T^{-1} = H$ and PHS$^\dagger$: $U_C H^* U_C^{-1} = -H$ become distinct from TRS and PHS, respectively. The 10 classes classified by TRS$^\dagger$, PHS$^\dagger$ and CS are called AZ$^\dagger$.
The most general classification in the non-Hermitian regime is based on the presence or absence of the following symmetries, which yields 38 classes called BL symmetry classes \cite{Bernard_2002,Kawabata_2019}:
\begin{align}
    &P \text{ symmetry} : &\mathcal{U}_P \mathcal{E} \mathcal{U}_P^{-1} = -\mathcal{E},\quad & \mathcal{U}_P^2 = \mathcal{I},\\
    &C \text{ symmetry} : &\mathcal{U}_C \mathcal{E}^\top \mathcal{U}_C^{-1} = \epsilon_C\mathcal{E},\quad & \mathcal{U}_C \mathcal{U}_C^* = \eta_C \mathcal{I},\\
    &K \text{ symmetry} : &\mathcal{U}_K \mathcal{E}^* \mathcal{U}_K^{-1} = \epsilon_K\mathcal{E},\quad & \mathcal{U}_K \mathcal{U}_K^* = \eta_K \mathcal{I},\\
    &Q \text{ symmetry} : &\mathcal{U}_Q \mathcal{E}^\dagger \mathcal{U}_Q^{-1} = \epsilon_Q\mathcal{E},\quad & \mathcal{U}_Q^2 = \mathcal{I},
\end{align} 
where $\mathcal{U}_\chi$ ($\chi=P,C,K,Q$) are unitary, $\epsilon_\chi = \pm1$ ($\chi=C,K,Q$), $\eta_\chi = \pm1$ ($\chi=C,K$), and $\mathcal{I}$ is the identity. AZ and AZ$^\dagger$ are subclasses of BL classes.

Since CPTP super-operators are non-Hermitian, the framework of the BL symmetry classes is applied to general quantum feedback control \cite{Nakagawa_2025}. If $\mathcal{E}$ exhibits an additional unitary symmetry, it can be block diagonalized with respect to the symmetry sectors until each block no longer has unitary symmetry. 
It suffices to analyze the BL symmetry of each block, where no further unitary symmetry is present.
In addition, we assume that each block has a steady state, and symmetry classes do not depend on $\tau_i \geq 0$, as also assumed in Ref. \cite{Nakagawa_2025}.
That is, any symmetry sector has an eigenvalue $\xi=1$, and its entire feedback process is characterized by a single symmetry class independent of the feedback durations. 

Now, we present the following theorem, which is the main result of this Letter.
\begin{thm}\label{main}
    Consider a family of CPTP maps of feedback control $\left\{ \mathcal{E}_{\tau_1,\cdots,\tau_N} \right\}$ in the form of Eqs.(\ref{general form of feedback control})-(\ref{unitary with feedback Hamiltonian}). If every measurement in the protocol is bare, then the possible BL symmetry classes of $\mathcal{E}_{\tau_1,\cdots,\tau_N}$ reduce to the ten-fold AZ$^\dagger$ subclass listed in Table \ref{tab:BL10}.
\end{thm}
\begin{table}
    \centering
    \caption{Ten classes admissible to successive quantum feedback controls with non-adaptive bare measurements. These are equivalent to AZ$^\dagger$ subclass of the BL symmetry classes if we multiply the original CPTP map by $i$. Here, 0 (1) represents that each symmetry is absent (present), and $\pm$ indicates the sign of $\eta_{C,K}$. We set $\epsilon_{C,K,Q} = +1$ in this table.}
    \label{tab:BL10}
    \begin{tabular}{lccc}
    \toprule
    Class & $C$ & $K$ & $Q$ \\
    \midrule
    $\mathrm A$ & 0 & 0 & 0 \\
    $\mathrm{psH}$ ($\mathrm{AIII}$) & 0 & 0 & 1 \\
    $\mathrm{AI}^{\dagger}$ & + & 0 & 0 \\
    $\mathrm{AI}+\mathrm{psH}^+$ ($\mathrm{BDI}^{\dagger}$) & + & + & 1 \\
    $\mathrm{AI}$ ($\mathrm D^{\dagger}$) & 0 & + & 0 \\
    $\mathrm{AI}+\mathrm{psH}^- $ ($\mathrm{DIII}^{\dagger}$) & $-$ & + & 1 \\
    $\mathrm{AII}^{\dagger}$ & $-$ & 0 & 0 \\
    $\mathrm{AII}+\mathrm{psH}^+$ ($\mathrm{CII}^{\dagger}$) & $-$ & $-$ & 1 \\
    $\mathrm{AII}$ ($\mathrm C^{\dagger}$) & 0 & $-$ & 0 \\
    $\mathrm{AII}+\mathrm{psH}^- $ ($\mathrm{CI}^{\dagger}$) & + & $-$ & 1 \\
    \bottomrule
    \end{tabular}
\end{table}
\noindent

\textit{Sketch of the proof.---}
We prove Theorem \ref{main} as follows.
First, from the assumption mentioned before, it is sufficient to consider the case of $\tau_i \to 0$, where the CPTP map is reduced to a sequence of bare measurements.
We can then utilize the following lemma, whose proof is given in the End Matter \ref{proof}. We note that the inequality part of the lemma is well-known \cite{sagawa_entropy_2022}, while the necessary condition for the equality might be nontrivial.

\begin{lemma}\label{1}
    The CPTP map of a bare measurement is contractive for Hilbert-Schmidt norm. That is, for every operator $\hat{\rho} \in \mathcal{L}(\mathcal{H})$,
    $||\hat{\rho}||_2 \geq ||\mathcal{E}(\hat{\rho})||_2$ holds for the CPTP map of a bare measurement $\mathcal{E}$ \cite{P_rez_Garc_a_2006}. Moreover, the equality holds if and only if $\hat{\rho} = \mathcal{E}(\hat{\rho})$.
\end{lemma}

Let $\mathcal{E}$ be a CPTP map of successive bare measurements and write $\mathcal{E}= \mathcal{E}^{(n)}\circ \cdots\circ \mathcal{E}^{(1)}$,
where $\mathcal{E}^{(i)}$ is the CPTP map of the $i$-th measurement.
Since each $\mathcal{E}^{(i)}$ is bare, the Hilbert-Schmidt norm of $\hat{\rho}$ does not increase under $\mathcal{E}^{(i)}$ from Lemma \ref{1}. Therefore, if $\mathcal{E}$ has an eigenvalue whose absolute value is 1, the induction shows $\mathcal{E}^{(i)}(\hat{\rho}) = \hat{\rho}$ for any $i$,
and hence $1$ is the only eigenvalue with unit absolute value.

On the other hand, 
we assumed that every symmetry sector has the eigenvalue $\xi = 1$.
This spectral structure forbids certain BL symmetries. Specifically, if $\mathcal{E}$ has an eigenvalue $\xi$ and has either $P$ symmetry or $C$ symmetry with $\epsilon_C = -1$, $-\xi$ is also an eigenvalue. Likewise, if $\mathcal{E}$ has $K$ or $Q$ symmetry, $\epsilon_K \xi^*$ or $\epsilon_Q\xi^*$ is also an eigenvalue. Consequently, successive bare measurements are incompatible with $P$ symmetry or any $\chi$ symmetry with $\epsilon_\chi = -1$ ($\chi = C,K,Q$). 
Taken together with some intrinsic constraints on the BL symmetries and the assumption that $\mathcal{E}$ has no residual unitary symmetry,
this result yields precisely the ten-fold AZ$^\dagger$ classes in Table \ref{tab:BL10} \cite{Nakagawa_2025}. 

The key in this proof is that the spectrum of successive bare measurements cannot respect certain symmetries. Therefore, it is expected that this result is robust against small perturbations or noise (i.e., some small unitaries after each bare measurement), because they are unlikely to induce those symmetries. For instance, in two-level systems, 
even if we can accurately engineer the $\mathcal{O}(\epsilon)$ perturbations, we need at least $N \gtrsim \mathcal{O}(\frac{1}{\sqrt{\epsilon}})$ measurements to produce $\xi = -1$ (see End Matter \ref{Zeno}).

We note that, in the above proof, we did not use the assumption that the feedback operation is unitary, but we only used the assumption that the symmetry class does not change in the limit of $\tau_i \to 0$.  Therefore, Theorem 1 is straightforwardly extended to the case that the feedback operation is dissipative and described by the Lindblad equations \cite{10.1063/1.522979,1976CMaPh..48..119L}.

\textit{Chiral Maxwell's demon with Gaussian errors.---}
We present an example of bare measurement-based feedback control.
We consider a chiral Maxwell's demon \cite{Nakagawa_2025} that cannot acquire precise information of the particle's position due to the measurement error.
A single particle hops on a one-dimensional lattice with $L$ sites, and the demon measures the particle's position and raises the potential at the left-neighboring site of the measurement outcome. We assume that the measurement has Gaussian errors; the measurement operator for outcome $m$ is given by \cite{Jacobs_2006}
\begin{equation}
    \hat{M}_m = \frac{1}{\mathcal{N}} \sum_{n=-c}^{c}e^{-\frac{\ell}{4}n^2}\ketbra{m-n},
\end{equation}
where $c$ is the cutoff of the error, $\ell$ controls the sharpness of the error, and $\mathcal{N}$ is the normalization factor. To keep the Bloch matrices finite-dimensional, we take the constant $c$ independent of the system size $L$.
After obtaining the outcome $m$, the demon applies feedback by raising the potential at site $m-1$. The corresponding feedback Hamiltonian is given by
\begin{equation}
    \hat{H}_m = -J \sum_{n}^{}(\ketbra{n}{n+1} + h.c.) + V\ketbra{m-1},
\end{equation}
where $J$ is the hopping amplitude and $V$ is the potential height raised by the demon. The Bloch matrices are 
\begin{align}
    X_{\mu\nu}(k) & = \frac{1}{L}\sum_{m}^{}\sum_{r,s}^{}(\hat{K}_m)_{r,s}(\hat{K}_m)^*_{r+\mu,s+\nu} e^{-ik(r-s)}\\
    & = \sum_{r}^{}\sum_{m}^{}(\hat{K}_m)_{r,0}(\hat{K}_m)^*_{r+\mu,\nu} e^{-ikr},
\end{align}
where we used the translational symmetry $(\hat{K}_m)_{r,s} = (\hat{K}_{m+i})_{r+i,s+i}$ to obtain the second line. The matrix elements of the Kraus operators are written as 
\begin{equation}
    \begin{aligned}
        (\hat{K}_m)_{r,s} & = \frac{1}{\mathcal{N}} \sum_{n=-c}^{c}e^{-\frac{\ell}{4}n^2}\bra{r}\hat{U}_m\ket{m-n}\bra{m-n}\ket{s}\\
        & = \frac{1}{\mathcal{N}}e^{-\frac{\ell}{4}(m-s)^2}\bra{r}\hat{U}_m\ket{s} \mathbbm{1}_{[-c+s,c+s]}(m),
    \end{aligned}
\end{equation}
where $\mathbbm{1}$ is the characteristic function, and hence the elements of the Bloch matrices are given by 
\begin{equation}
    \begin{aligned}
        X_{\mu\nu}(k) = & \sum_{r}^{}\sum_{m}^{}\frac{\mathbbm{1}_{[-c,c]\cap[-c+\nu,c+\nu]}(m)}{\mathcal{N}^2}e^{-\frac{\ell(m^2+(m-\nu)^2)}{4}}\\
         & \hspace{1.2cm}\times \bra{r}\hat{U}_m\ketbra{0}{\nu}\hat{U}_m^\dagger\ket{r+\mu} e^{-ikr}.
    \end{aligned}
\end{equation}
Therefore, $X_{\mu\nu}(k)$ has nonzero entries only in $4c+1$ columns with $-2c\leq\nu\leq2c$, and hence can be truncated to a $4c+1$ dimensional square matrix $X_{\mathrm{trunc}}(k)$. The winding number around the base-point $\xi_{\text{PG}} \in \mathbb{C}$ is defined by $w(\xi_{\mathrm{PG}}) = \frac{1}{2\pi i} \int_{-\pi}^{\pi} \partial_k \log{\det(X_{\mathrm{trunc}}(k)-\xi_{\mathrm{PG}})} dk$.

\begin{figure}
    \includegraphics[width=\linewidth]{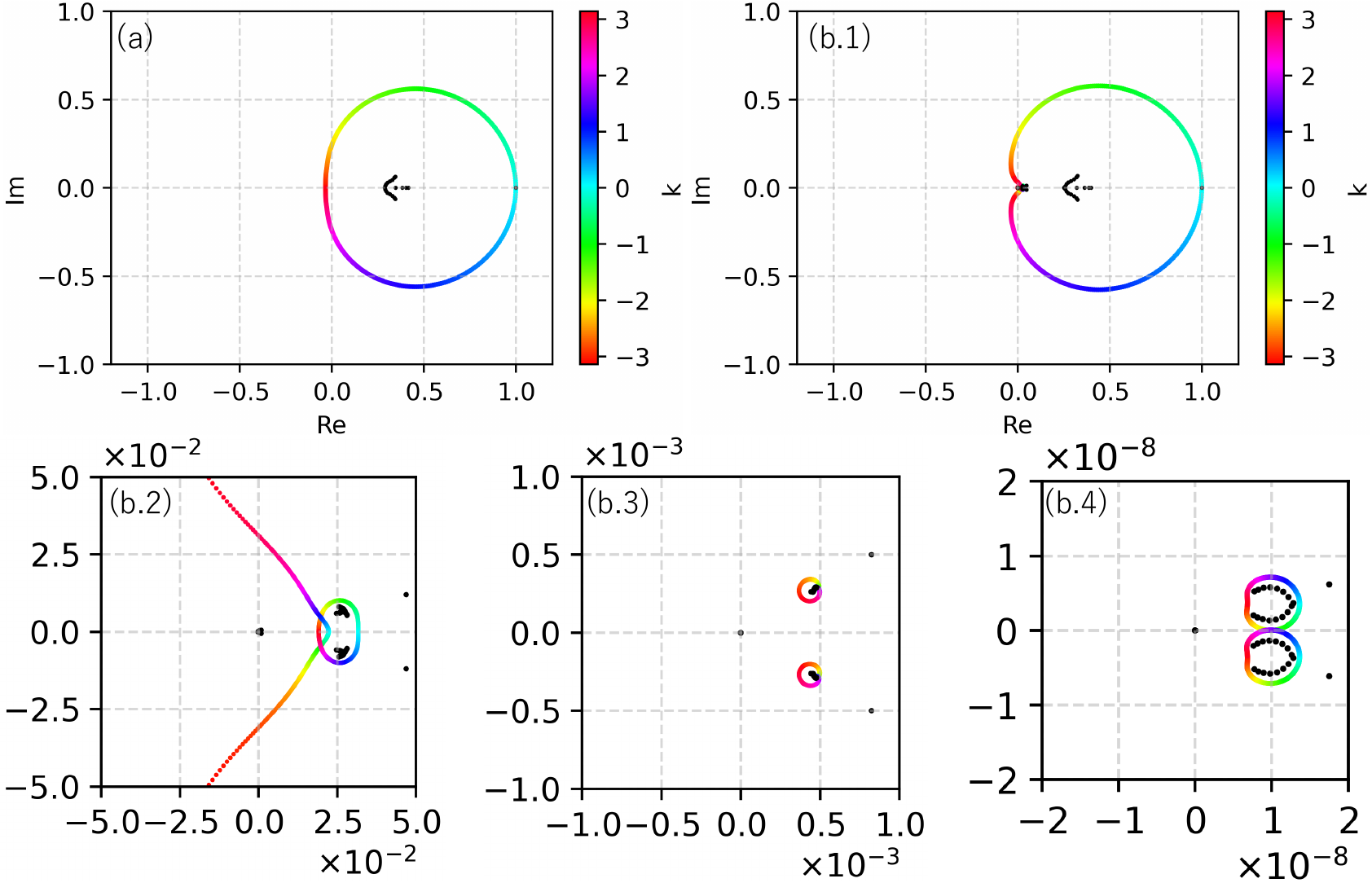}  
    \caption{Spectra of the CPTP maps of chiral Maxwell's demon with projective measurement (a) and that with Gaussian errors (b.1) -- (b.4). The colored dots are the eigenvalues under PBC for size $L=400$, whose colors correspond to the wavenumber $k$ of the Bloch matrices, and the black dots are the eigenvalues under OBC for size $L = 20$.
The unit of energy, the height of the potential and the feedback duration are set to be $J=1,V=10^2$, $\tau=1$ respectively. We set the radius of the error (cutoff) to be $c=2$, and the Gaussian sharpness is set to be $\ell=15$. The winding number $w$ around the origin is $w(\xi_{\text{PG}}=0)=-1$ for the left and $w(\xi_{\text{PG}}=0)=0$ for the right, and that around $1-\epsilon$ is $w(\xi_{\text{PG}}=1-\epsilon)=-1$ for both pictures with $\epsilon = 10^{-12}$.}
\label{spectrum_of_chiralMD_with_gaussian_error}
\end{figure}

Figure \ref{spectrum_of_chiralMD_with_gaussian_error} shows the eigenspectrum of the CPTP map $\mathcal{E}$ of chiral Maxwell's demon with Gaussian errors.
The result is for finite system size $L$, but we observe the robustness of the spectrum against an increase in $L$ (see End Matter \ref{systemsize}), which justifies the extrapolation to $L\to\infty$.
We can observe that the high degeneracy at $\xi=0$ is partially broken. This is because the error of the measurement keeps the off-diagonal elements $\ketbra{i}{j}$ with $|i-j| \leq 2c$ nonzero, leaving $4c+1$ bands. If we do not take the cutoff $c$ or let $c$ be proportional to the system size $L$, an infinite number of bands arise (i.e., $\rank X(k) \to \infty$), making the winding number ill-defined (see also End Matter \ref{behavior_c}). This is consistent with the framework in \cite{Nakagawa_2025}, where the locality of the measurement and feedback is assumed.
Notably, these additional bands around the origin change the winding number around $\xi_{\text{PG}}=0$. 
By contrast, the spectral structure around $\xi=1$ is insensitive to the measurement error, which motivates us to alternatively choose the base-point at $\xi_{\text{PG}} = 1-\epsilon$ with sufficiently small $\epsilon > 0$. The reason why we need $\epsilon$ is that the spectrum of any CPTP map exactly passes $1$ due to the trace-preserving condition, while the base-point should be avoided by the spectrum when defining the winding number.
Another motivation for $\xi_{\text{PG}} = 1-\epsilon$ is the analogy with the topology of one-dimensional classical stochastic processes \cite{Murugan2017,sawada2024bulkboundarycorrespondenceergodicnonergodic,Sawada_2024}, where the natural choice of the base-point corresponds to the steady state.

We remark that the extra bands near the origin have a quantum origin; some off-diagonal degrees of freedom survive due to measurement error and contribute to the new bands. This cannot happen with classical stochastic dynamics because it does not make off-diagonal elements alive, leaving only one band in the complex plane. This sensitivity to off-diagonal elements
may be exploited to elucidate the role of quantum coherence in quantum feedback control.

As shown by black dots in Fig. \ref{spectrum_of_chiralMD_with_gaussian_error}, the loop structure of the spectrum under PBC collapses under the open boundary condition (OBC) for both the cases with and without errors. This implies that the non-Hermitian skin effect \cite{PhysRevX.8.031079,Okuma_2023,Okuma_2020,PhysRevLett.125.126402} of quantum feedback control \cite{Nakagawa_2025} persists under measurement errors. 

We also note that this model can be applied to the continuous measurement regime by taking $\ell$ proportional to $\tau$ and letting $\tau\to0$ \cite{Jacobs_2006}. However, from this perspective it is no longer possible to appropriately take the cutoff $c$ of the Gaussian tail (i.e., the measurement cannot be local anymore), and therefore, the winding number might not be well-defined since the size of the Bloch matrices diverges as $L\to\infty$, accompanied by an infinite number of bands scattered on the complex plane. Nevertheless, based on the observed behavior of the spectrum against $c$ (see also End Matter \ref{behavior_c}), we expect that those infinite bands are condensed around $0$. If the only accumulation point of the spectrum is $0$, we are still able to define the winding number around $\xi_{\text{PG}} = 1 - \epsilon$.

\textit{Feedback control beyond the ten-fold symmetry.---}
We have shown that successive feedback control with bare measurements cannot realize symmetry classes outside the AZ$^\dagger$ classes. Now, we show that some forbidden classes are accessible with non-bare measurements.

Our strategy is to mix the spin-up space and spin-down space to produce $\xi = -1$ within the measurement process. We consider spin-flipping measurement whose Kraus operators are $\hat{M}_{j,\uparrow} = \ketbra{j,\downarrow}{j,\uparrow},\quad \hat{M}_{j,\downarrow} = \ketbra{j,\uparrow}{j,\downarrow}$.
This measurement can be implemented by the projective measurement of position and spin followed by an instantaneous $\pi$ pulse of magnetic field. 
We use this measurement to construct a protocol that has $P$ symmetry as follows.
\begin{enumerate}
    \item Perform the spin-flipping measurement $\{ \hat{M}_{j,\uparrow}, \hat{M}_{j,\downarrow} \}$, and denote the outcome as $m_1$.
    \item Apply a feedback $\hat{U}_{m_1}(\tau)$ conditioned on $m_1$, given by a feedback Hamiltonian via Eq. (\ref{unitary with feedback Hamiltonian}).
    \item Perform the projective measurement with $\left\{ \ket{j,\sigma} \right\}$.
\end{enumerate}

The Kraus operators of the corresponding CPTP map $\mathcal{E}(\hat{\rho}) = \sum_{m_1,m_2}^{}\hat{K}_{m_1,m_2}\hat{\rho}\hat{K}_{m_1,m_2}^\dagger$ are given by
\begin{align}
    \hat{K}_{m_1,m_2} = \ketbra{m_2}\hat{U}_{m_1}(\tau)\ketbra{\overline{m_1}}{m_1},
\end{align}
where $\overline{m} = (j,\uparrow)$ and $(j,\downarrow)$ are defined for $m = (j,\downarrow)$ and $(j,\uparrow)$, respectively. 
Collecting the non-zero blocks of the Bloch matrices defined in Eq.(\ref{Bloch Matrix}), we get 
\begin{gather}
    X_{\text{trunc}}(k)= \begin{pmatrix} \xi_{\uparrow\downarrow}(k) & \xi_{\uparrow\uparrow}(k)\\ \xi_{\downarrow\downarrow}(k) & \xi_{\downarrow\uparrow}(k) \end{pmatrix},\\
    \qquad \xi_{\sigma,\sigma'}(k)=\frac{1}{L}\sum_{j,j'} p(j',\sigma|j,\sigma')\,e^{-ik(R_{j'}-R_j)},
\end{gather}
where $p(m_2|\overline{m_1})$ is the transition probability under the unitary $\hat{U}_{m_1}(\tau)$. If the spin is conserved under the feedback Hamiltonian, $\xi_{\uparrow\downarrow} = \xi_{\downarrow\uparrow} = 0$ holds, and the CPTP map satisfies the following $P$ symmetry:
\begin{equation}
    \sigma_z X_{\text{trunc}}(k) \sigma_z = - X_{\mathrm{trunc}}(k).
\end{equation}
Therefore, $P$ symmetric feedback control beyond the AZ$^\dagger$ classes is indeed achievable. 

We note that the above measurement is not bare but unital. Therefore, unitality of measurements is not enough to exclude the symmetry classes outside the AZ$^\dagger$ classes. In addition, in order to implement symmetry classes outside the AZ$^\dagger$ with a single measurement protocol, we have to make all Kraus operators traceless (see End Matter \ref{beyond}). This requirement is consistent with the above example, where the spin flips make the Kraus operators traceless.

\textit{Conclusions and discussions.---}
We proved that the symmetry classes accessible to a broad class of quantum feedback control reduce to the ten-fold AZ$^\dagger$ classes out of the 38-fold BL classes. 
Importantly, this restriction stems solely from the spectral structure of the measurements.
For instance, as we discussed above, the AZ$^\dagger$ limitation can be lifted by imposing additional symmetries on the spectrum of the measurement with non-bare measurements. 

Symmetry is essential in classifying topological phases, and therefore forms a core design principle for engineering topological quantum feedback control. Our result provides such a guiding principle for a broad class of protocols, especially those with measurement error. 
The chiral Maxwell's demon with Gaussian errors would serve as a practical model in the experimental implementation by, for example, ultracold atoms with single-site addressing techniques \cite{Bakr2009Nature,Sherson2010Nature,Endres2016Science}.

We conclude this Letter with some future prospects.
First, it remains to be discussed whether the AZ$^\dagger$ restriction holds when one allows adaptive measurements, where one can decide what measurement to perform depending on the earlier outcomes.
Moreover, considering that some measurement processes, such as photon counting, inevitably possess a unitary component followed by the bare part \cite{Wiseman_Milburn_2009}, it is essential to discuss what happens beyond the bare measurement assumption.
It is still unclear whether the CPTP condition alone casts some intrinsic constraint on the classification.
Here, we \textit{conjecture} that all the 38 BL symmetry classes are accessible by engineering the unitary parts of non-bare measurements. 

Application of our framework to continuous measurement and feedback is another future perspective. As we discussed above, chiral Maxwell's demon with Gaussian errors can be applied to the continuous measurement regime by taking an appropriate limit of parameters, although its topological aspects are not yet well clarified. 
Elucidating the topological structure in this regime is a key step toward establishing robustness for real-time control and exploring practical design rules for, e.g., continuous error mitigation and cooling.\\

\textit{Acknowledgements.---}
We are grateful to Masaya Nakagawa, Atsushi Oyaizu, Takeshi Fukuhara, Shoki Sugimoto, and Kaito Tojo for valuable discussions.
This work was supported by JST ERATO Grant Number JPMJER2302, Japan. T.S. is also supported by JST CREST Grant No. JPMJCR20C1 and by Institute of AI and Beyond of the University of Tokyo. Z.G. acknowledges support from the University of Tokyo Excellent Young Researcher Program.

\textit{Data availability.---}
The data that support the findings of this article are openly available \cite{wen_2025_17276238}.

\bibliography{Symmetry_of_Feedback_Control}

\appendix
\section*{End Matter}
\renewcommand{\thelemma}{\Alph{lemma}}
\setcounter{lemma}{0}
\renewcommand{\thesubsection}{E\arabic{subsection}}
\subsection{Proof of Lemma 1}\label{proof}
Let $\mathcal{E}$ be a CPTP map of a bare measurement on Hilbert space $\mathcal{H}$. 
Since the inequality part of the lemma is well-known \cite{sagawa_entropy_2022}, here we only prove that for a linear operator $\hat{\rho} \in \mathcal{L}(\mathcal{H})$
\begin{equation}
    || \mathcal{E}(\hat{\rho}) ||_2 = ||\hat{\rho}||_2 \Longleftrightarrow \mathcal{E}(\hat{\rho}) = \hat{\rho}.
\end{equation}

The implication``$\Leftarrow$'' is trivial, so we focus on ''$\Rightarrow$''. we first prove the following lemma.
\begin{lemma}\label{A}
    Let $\mathcal{E}$ be a unital CPTP map and $\{ \hat{M}_m \}$ be its Kraus operators. Then, 
    \begin{equation}
        ||\hat{\rho}||_2^2 - ||\mathcal{E}(\hat{\rho})||_2^2 = \frac{1}{2}\sum_{\alpha,\beta}^{}\Big|\Big|\left[ \hat{\rho}, \hat{M}_\alpha^\dag\hat{M}_\beta \right]\Big|\Big|_2^2.
    \end{equation}
\end{lemma}
\begin{proof}
    We prove this lemma by a straightforward calculation;
    \begin{equation}
        \begin{aligned}
            & \sum_{\alpha,\beta}^{}\Big|\Big|\left[ \hat{\rho}, \hat{M}_\alpha^\dagger\hat{M}_\beta \right]\Big|\Big|_2^2\\
        = & \sum_{\alpha,\beta}^{} ||\hat{\rho}\hat{M}_\alpha^\dagger\hat{M}_\beta - \hat{M}_\alpha^\dagger\hat{M}_\beta\hat{\rho}||_2^2\\
        = &  \sum_{\alpha,\beta}^{} \Tr\left( 
            (-\hat{\rho}^\dagger \hat{M}_\beta^\dagger\hat{M}_\alpha + \hat{M}_\beta^\dagger\hat{M}_\alpha\hat{\rho}^\dagger)
            (\hat{\rho}\hat{M}_\alpha^\dagger\hat{M}_\beta - \hat{M}_\alpha^\dagger\hat{M}_\beta\hat{\rho})
             \right)\\
        = &  \sum_{\alpha,\beta}^{} \Tr\left(-\hat{\rho}^\dagger \hat{M}_\beta^\dagger\hat{M}_\alpha\hat{\rho}\hat{M}_\alpha^\dagger\hat{M}_\beta
         + \hat{\rho}^\dagger \hat{M}_\beta^\dagger\hat{M}_\alpha\hat{M}_\alpha^\dagger\hat{M}_\beta\hat{\rho} \right.\\
        & \qquad\qquad \left. + \hat{M}_\beta^\dagger\hat{M}_\alpha\hat{\rho}^\dagger\hat{\rho}\hat{M}_\alpha^\dagger\hat{M}_\beta
         - \hat{M}_\beta^\dagger\hat{M}_\alpha\hat{\rho}^\dagger\hat{M}_\alpha^\dagger\hat{M}_\beta\hat{\rho}\right)\\
        = &  2\Tr(\hat{\rho}^\dagger\hat{\rho}) - 2 \sum_{\alpha,\beta}^{}\Tr(\hat{M}_\alpha\hat{\rho}^\dagger\hat{M}_\alpha^\dagger\hat{M}_\beta\hat{\rho}\hat{M}_\beta^\dagger)\\
        = &  2 \left( ||\hat{\rho}||_2^2 - ||\mathcal{E}(\hat{\rho})||_2^2 \right).
        \end{aligned}
    \end{equation}
    Here, we used the condition that $\mathcal{E}$ is unital (i.e. $\sum_{\alpha}^{}\hat{M}_\alpha\hat{M}_\alpha^\dagger = \hat{I}$) and the trace preserving property of CPTP maps (i.e. $\sum_{\alpha}^{}\hat{M}_\alpha^\dagger\hat{M}_\alpha = \hat{I}$).
\end{proof}
We use Lemma \ref{A} to prove Lemma \ref{1}. 
Since Kraus operators of $\mathcal{E}$ are positive, $\mathcal{E}$ can be written as 
\begin{equation}
    \mathcal{E}(\hat{\rho}) = \sum_{\alpha}^{}\sqrt{\hat{E}_\alpha} \hat{\rho} \sqrt{\hat{E}_\alpha},
\end{equation}
where $\left\{ \hat{E}_\alpha \right\}$ is the POVM of the bare measurements.
It is straightforward to confirm that CPTP maps of bare measurements are unital. Therefore, for any $\alpha,\beta$, we have $\left[ \hat{\rho}, \sqrt{\hat{E}_\alpha}\sqrt{\hat{E}_\beta} \right] = 0$. Especially, 
\begin{equation}
    \left[ \hat{\rho}, \hat{E}_\alpha \right] = 0
\end{equation}
for any $\alpha$.

Next, let us fix $\alpha$ and consider the basis diagonalizing $\hat{E}_\alpha$. By denoting the matrix elements of $\hat{E}_\alpha$ and $\hat{\rho}$ as $\lambda_i \delta_{ij}$ and $\rho_{ij}$ respectively, we get
\begin{equation}
    \lambda_i \rho_{ij} = \lambda_j \rho_{ij}.
\end{equation}
Since $\lambda_i \geq 0$, this is equivalent to 
\begin{equation}
    \sqrt{\lambda_i}\rho_{ij} = \sqrt{\lambda_j}\rho_{ij}.
\end{equation}
This means that $\left[ \hat{\rho},\sqrt{\hat{E}_\alpha} \right] = 0$ for any $\alpha$.
Therefore,
\begin{equation}
    \begin{aligned}
        \mathcal{E}(\hat{\rho}) & = \sum_{\alpha}^{}\sqrt{\hat{E}_\alpha}\hat{\rho}\sqrt{\hat{E}_\alpha}\\
        & = \sum_{\alpha}^{}\hat{E}_\alpha\hat{\rho} = \hat{\rho}.
    \end{aligned}
\end{equation}
Here, we used $\sum_{\alpha}^{}E_\alpha = \hat{I}$ in the last equality. Thus, Lemma \ref{1} is proved.

\subsection{Beyond the AZ$^\dag$ with a single measurement protocol}\label{beyond}
We show that in order to obtain a protocol beyond the AZ$^\dagger$ classes with a single measurement, we have to make all Kraus operators of the measurement traceless.

First, we note again that it suffices to consider the case of $\tau\to0$, and hence we consider the CPTP map of the single measurement $\mathcal{E}(\hat{\rho}) = \sum_{m}^{} \hat{M}_m \hat{\rho} \hat{M}_m^\dagger$.
The symmetries excluded in AZ$^\dagger$ are 
\begin{align}
    \mathcal{U}_P \mathcal{E} \mathcal{U}_P^{-1} = -\mathcal{E}, \label{P}\\
    \mathcal{U}_C \mathcal{E}^\top \mathcal{U}_C^{-1} = -\mathcal{E},\label{C-}\\
    \mathcal{U}_K \mathcal{E}^* \mathcal{U}_K^{-1} = -\mathcal{E},\label{K-}\\
    \mathcal{U}_Q \mathcal{E}^\dagger \mathcal{U}_Q^{-1} = -\mathcal{E}.\label{Q-}
\end{align}
Meanwhile, the trace of the super-operator $\mathcal{E}$ is given by 
\begin{equation}\label{TrE}
    \Tr \left[ \mathcal{E} \right] = \Tr \left[ \sum_{m}^{} \hat{M}_m \otimes \hat{M}_m^* \right] = \sum_{m}^{}\abs{\Tr \left[ \hat{M}_m \right]}^2,
\end{equation}
where the first equality follows from the Choi-Jamiolkowski isomorphism \cite{JAMIOLKOWSKI1972275}. Taking the trace of both sides of Eqs. (\ref{P}) - (\ref{Q-}), one gets $\Tr\left[  \mathcal{E} \right] = -\Tr \left[ \mathcal{E} \right]$ for Eqs. (\ref{P}),(\ref{C-}), and $\Tr \left[ \mathcal{E}^* \right] = -\Tr \left[ \mathcal{E} \right]$ for Eqs. (\ref{K-}),(\ref{Q-}).
Since $\Tr \left[ \mathcal{E} \right] \in \mathbb{R}$ (by Eq. (\ref{TrE})), $\mathcal{E}$ must be traceless for all symmetries in Eqs. (\ref{P}) - (\ref{Q-}). 
Using Eq. (\ref{TrE}) again, we find that in order for $\mathcal{E}$ to be traceless, all Kraus operators must be traceless. 

Notably, this condition is not achievable with bare measurements, whose Kraus operators are all positive. This provides an alternative proof of Theorem 1, but only for the case with a single measurement. 

\subsection{Eigenvalues of Successive Projective Measurements in Two-level Systems}\label{Zeno}
We showed that an eigenvalue $\lambda$ of successive bare measurements cannot reach $-1$ in the proof of Theorem \ref{main}.
On the other hand, $\lambda$ can take any negative value with $|\lambda|<1$.
We show this in the simplest setup; successive projective measurements on the two-dimensional Hilbert space $\mathcal{H}$.

In two-level systems, $\mathcal{L}(\mathcal{H})$ is spanned by an orthonormal basis $\left\{ \hat{I}, \hat{\sigma}_x, \hat{\sigma}_y, \hat{\sigma}_z \right\}$ with regard to the Hilbert-Schmidt inner product. Since an eigenmode of a CPTP map with eigenvalue $\lambda\neq1$ is always traceless, it suffices to consider the traceless subspace of $\mathcal{L}(\mathcal{H})$, which is spanned by $\left\{ \hat{\sigma}_x,\hat{\sigma}_y,\hat{\sigma}_z \right\}$.
We represent $ \hat{\rho} = \bm{v}\cdot\bm{\hat{\sigma}} \in \mathrm{span}\left\{ \hat{\sigma}_x,\hat{\sigma}_y,\hat{\sigma}_z \right\}$ as a three-dimensional vector $\bm{v}\in\mathbb{C}^3$. 
Since projective measurements always map $\hat{I}$ to $\hat{I}$, a projective measurement on $\bm{v}\cdot\bm{\hat{\sigma}}$ behaves in the same way as that on $\frac{1}{2}(\hat{I} + \bm{v}\cdot\bm{\hat{\sigma}})$ in the Bloch sphere.

Next, we consider the projective measurement with the basis $\ket{\psi_+} = \cos{\frac{\theta}{2}} \ket{0} + e^{i\phi}\sin{\frac{\theta}{2}}\ket{1}$ and $\ket{\psi_-} = \sin{\frac{\theta}{2}}\ket{0} - e^{i\phi}\cos{\frac{\theta}{2}}$, and let $\bm{v}_+ = (\sin\theta\cos\phi,\sin\theta\sin\phi,\cos\theta)$, $\bm{v}_- = -\bm{v}_+$ be the points corresponding to $\ket{\psi_\pm}$ in the Bloch sphere.
This projective measurement on $\hat{\rho}$ is represented as the orthogonal projection to the one dimensional space spanned by $\bm{v}_+$.
Therefore, if we perform  the $i$-th projective measurement with $\theta_i = \frac{\pi}{N}$ and $\phi_i=0$ for the initial vector $\hat{\rho} = \hat{\sigma}_z$, the $N$ steps of successive projective measurements map it as $\hat{\sigma}_z \mapsto -(\cos{\frac{\pi}{N}})^N \hat{\sigma_z}$. By taking the limit $N\to\infty$, the eigenvalue approaches $\lim_{N\to\infty}-(\cos{\frac{\pi}{N}})^N = -1$, and therefore we can achieve any negative eigenvalue larger than $-1$. 
In this setting, we can reach  $\xi = -1 + \mathcal{O}(N^{-1})$ with $N$ measurements. Therefore, according to the Weyl's inequality, in order to produce $\xi = -1$ with $\mathcal{O}(\epsilon)$ perturbations to each measurement, $N$ should satisfy $N\epsilon \gtrsim \frac{1}{N}$, which means that we need a large number of measurements.
We note that this result can be interpreted as an example of the quantum Zeno effect \cite{vonNeumann2018}.

\subsection{Detailed Behavior of Chiral Maxwell's Demon with Gaussian Errors}

\begin{figure}[H]
    \includegraphics[width=1\linewidth]{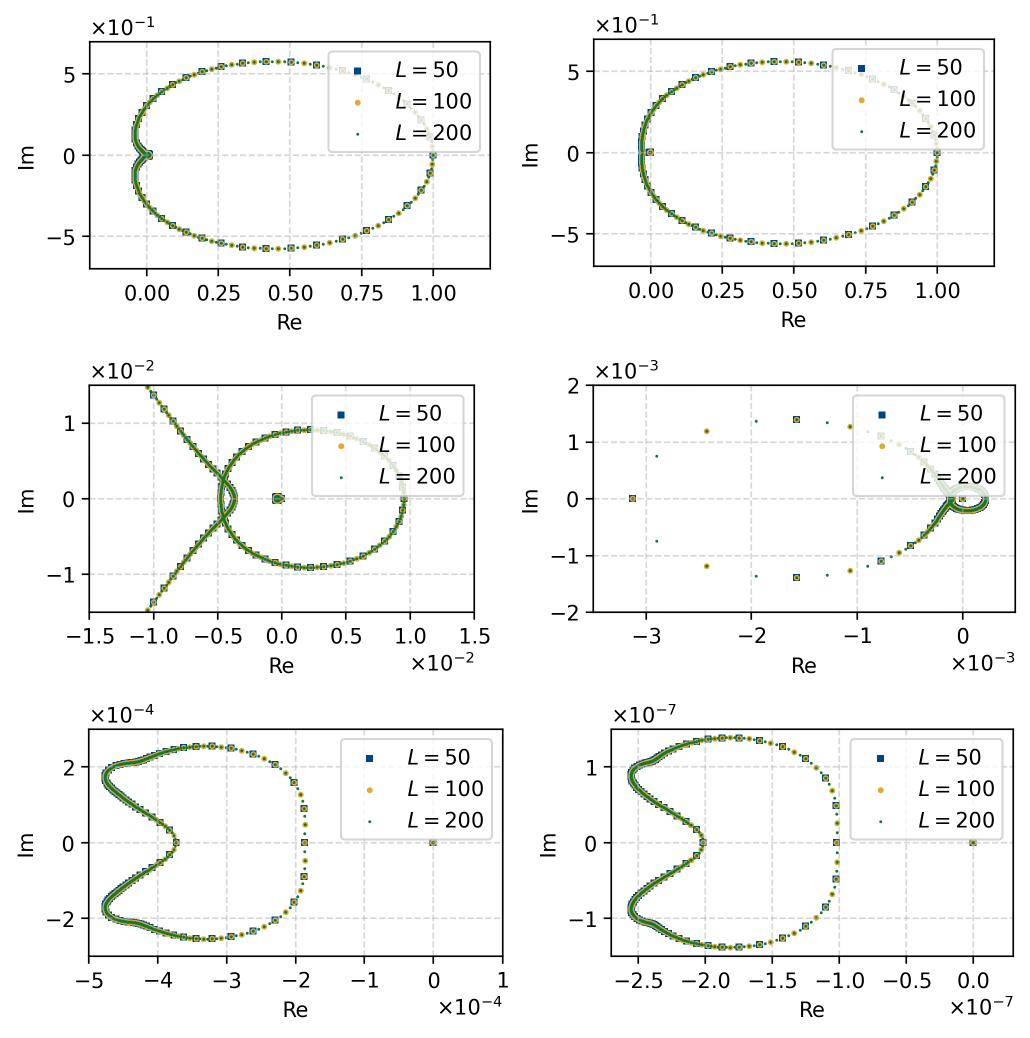}
    \caption{The PBC spectra of the CPTP maps of chiral Maxwell's demon with Gaussian errors for system size $L=50, 100, 200$. The Gaussian sharpness is taken to be $\ell=15$ $(30)$ for the left (right). The color of PBC spectrum corresponds to the wavenumber $k$ of the Bloch matrices. The unit of energy, the height of the potential and the feedback duration are set to be $J=1, V=10^4$, $\tau=1$, respectively. We set the radius of the error (cutoff) to $c=2$.}
    \label{insensitive_to_L}
\end{figure}
\begin{figure}
    \centering
    \includegraphics[width=0.98\linewidth]{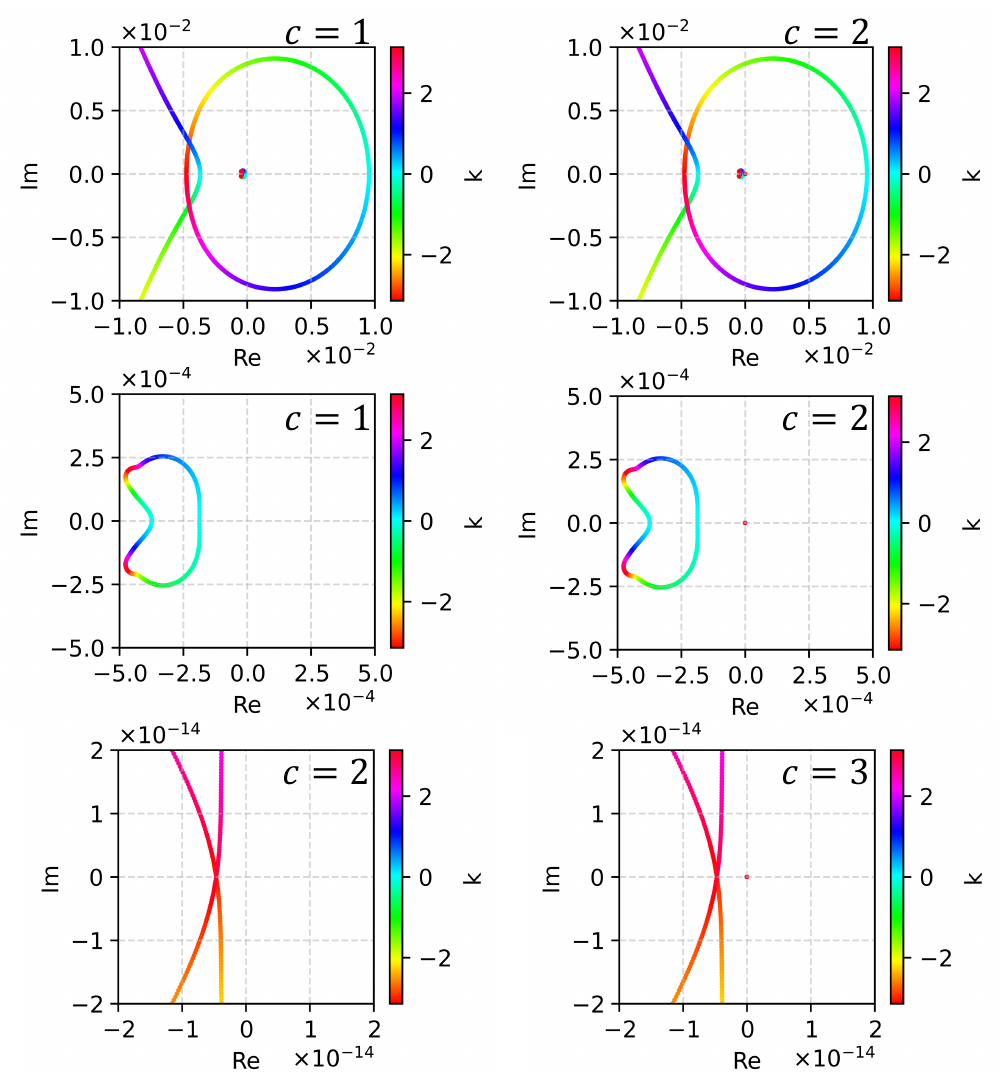}
    \caption{The dependence of PBC spectra against the Gaussian tail cutoff $c$ of the measurement error in chiral Maxwell's demon. The colored dots are the eigenvalues under PBC for $L=400$, whose colors correspond to the wavenumber $k$ of the Bloch matrices. The unit of energy, the height of the potential and the feedback duration are set to be $J=1, V=10^4$, $\tau=1$, respectively. The Gaussian sharpness is taken to be $\ell=15$.}
    \label{cutoff_dependence}
\end{figure}

\subsubsection{Insensitivity of PBC spectrum to the system size}\label{systemsize}
In Fig. \ref{insensitive_to_L}, the PBC spectra for $L=50,100,200$ are displayed. All eigenvalues for $L=200$ overlap with those for $L=100$, and all eigenvalues for $L=100$ overlap with those for $L=50$. This serves as numerical evidence that the spectrum is not sensitive to the system size $L$, justifying the speculation that the spectrum for $L=\infty$ passes all eigenvalues for finite $L$.

\subsubsection{Behavior against Gaussian tail cutoff}\label{behavior_c}

The spectral change with respect to the Gaussian tail cutoff $c$ of the measurement error is focused on in Fig. \ref{cutoff_dependence}. Except for the difference between the outermost band for $c=0$ and $c=1$, the increase in $c$ just adds some bands inside the formerly innermost band and does not drastically change the outer bands. This observation leads to the conjecture that, in the limit of $c\to\infty$, $\xi=0$ is the only accumulation point of the spectrum of chiral Maxwell's demon with Gaussian errors, allowing the winding number around $1-\epsilon$ to converge in the limit of continuous measurement and feedback.

\nocite{wen_2025_17276238}

\end{document}